\documentclass [11pt] {article}
\usepackage{amsmath, amsthm, amssymb,latexsym,amscd}
\usepackage{graphics}
\usepackage[usenames]{color}
\usepackage{geometry}
\makeatletter

\usepackage[T1]{fontenc}
\usepackage[dvips]{graphicx}
\usepackage{rotating,amssymb,amsmath,amsfonts,delarray}
\usepackage[usenames]{color}
\usepackage{placeins}
\usepackage{float}
\newtheorem{thm}{Theorem}[section]

\newtheorem{cor}{Corollary}[section]

\newtheorem{defn}{Definition}[section]

\newtheorem{pro}{Proposition}[section]

\title {On the Capital Allocation Problem for a New Coherent Risk Measure in Collective Risk Theory}

\author{Hirbod Assa \footnote{Mailing address: Hirbod Assa. Institute for Financial and Actuarial Mathematics. University of Liverpool. UK. Email: assa@liverpool.ac.uk}\\
University of Liverpool \and
Manuel Morales \footnote{Mailing address: Manuel Morales. Department of Mathematics and Statistics. University of Montreal. CP. 6128 succ. centre-ville. Montreal, Quebec. H3C 3J7. CANADA. Email: morales@dms.umontreal.ca}\\
University of Montreal \and
Hassan Omidi Firouzi \footnote{Mailing address: Hassan Omidi Firouzi. Department of Mathematics and Statistics. University of Montreal. CP. 6128 succ. centre-ville. Montreal, Quebec. H3C 3J7. CANADA. Email: omidifh@dms.umontreal.ca}\\
University of Montreal 
}

\date{{\scriptsize First draft: February~8, 2013. This version: \today.}}
\begin{document}
\maketitle
\begin{abstract}
In this paper we introduce a new coherent cumulative risk measure on  $\mathcal{R}_L^p$, the space of c\`adl\`ag processes having Laplace transform. This new coherent risk measure turns out to be tractable enough within a class of models where the aggregate claims is driven by a spectrally positive L\'evy process. Moreover, we study the problem of capital allocation in an insurance context and we show that the capital allocation problem for this risk measure has a unique solution determined by the Euler allocation method. Some examples are provided. \\

 \end{abstract}
 
\section{Introduction}
Collective risk theory has built upon the pioneering work of Filip Lundberg \cite{Cramer} and it now comprises a substantial body of knowledge that concerns itself with the study the riskiness of an insurer's reserve as measured by the ruin probability and related quantities \cite{Asmussen}. A large amount of literature now exists on such insolvency measures for a wide variety of models, the latest being the so-called L\' evy insurance risk models \cite{Biffis_Morales} and \cite{Biffis_Kyprianou}.
 
Traditionally, risk theory focuses on the insurer's ability to manage the solvency of its reserve through the control of initial investment $x$. The mathematical tool often cited for such task is the probability of ruin since it is a measure of the likelihood that an insurer's reserve would eventually be insufficient to cover its liabilities in the long run. 

More precisely, consider the following general model for the risk reserve of an insurance company,
\begin{equation}\label{risk}
R(t) = x + c\,t - X(t) \;,\quad t \geq 0\;,
\end{equation}
where the aggregate claims process $X$ is a spectrally positive L\'evy process with zero drift, with $X(0)=0$ and jump measure denoted by $\nu$. Moreover, $x$ is the initial reserve level and $c$ is a constant premium rate defined as
\begin{equation}\label{Loading}
c = (1+\theta) \mathbb{E}[X(1)]
\end{equation}
  where $\theta > 0$ is the security loading factor.

Then the associated ruin time is 
\begin{equation}\label{ruin_time}
\tau_x := \inf\{t \geq 0 \;|\; X(t) - c\,t \geq x \} \;,
\end{equation}
and the {\it infinite-horizon ruin probability} can be defined by
\begin{equation}\label{ruin}
 \psi(x) := \mathbb{P}^x(\tau_x<\infty) \;,
\end{equation}
where $\mathbb{P}^x$ is short-hand notation for $\mathbb P( \;\cdot \;|\; X(0)=x)$. 

Much of the literature in collective risk theory studies the problem of deriving expressions and reasonable approximations for the probability of ruin as a function of the initial reserve level $x$. This problem is addressed within an ever-growing set of models for the aggregate claims process. See \cite{Asmussen} for a thorough account on the so-called ruin theory. 

Naturally, the ruin probability $\psi$ quantifies the solvency of the net-loss process $Y_t:=X_t-ct$ as a function of the initial reserve level $x$. In fact, we can define a risk measure $\rho_{\beta} : \mathcal X \longrightarrow [0,1]$ on a suitable model space $\mathcal X$ (say the space of bounded c\`adl\`ag stochastic processes $\mathcal R^{\infty}$). Let $Y_t=ct-X_t$ be the net-loss process associated with the reserve process (\ref{risk}), then 
\begin{equation}\label{ruin_measure}
\rho_{\beta}(Y)\longmapsto a := \inf \{x \geq 0 \;| \; \psi(x) \leq \beta\}\;,
\end{equation}
where $\psi$ is the associated ruin probability (\ref{ruin}) and $\beta \in [0,1]$ represents a given tolerance to ruin. 

One can interpret $a$ as the smallest initial level for which the process $R$ has an acceptable risk level, i.e. its associated ruin probability is less or equal to a tolerable figure $\beta$. Such risk measures have been recently studied (see \cite{Trufin}) and although they exhibit interesting properties, they lack the tractability of an efficient risk management tool. In fact, any meaningful risk management application, such as capital allocation, would be hard to implement using (\ref{ruin_measure}).

In this paper, we recover this idea of measuring the risk of an insurance risk process and we define a coherent risk measure on the space of c\`adl\`ag processes having Laplace transform as a mapping $\rho:\mathcal{R}_L^p \longrightarrow \mathbb R_+$. Unlike (\ref{ruin_measure}), this measure is tractable enough and allows for a solution of the capital allocation problem. This is carried out within the framework given by the theory of coherent and convex risk measures defined on a suitable space of stochastic processes. Among previous works on these issues we find \cite{Cheridito1} and \cite{Cheridito2} where the authors work out risk measures on the space of random processes modeling the outcome of a certain financial position and \cite{Cheridito3} where they develop risk measures in a dynamic fashion.   

The contribution of this paper is two-fold. First, based on \cite{Ahmadi} and \cite{Assa2}, we design a new risk measure on the space of bounded c\` adl\`ag processes that can capture the risk associated with the path-properties of an insurance model. We do this by extending the notion of {\it Entropic Value at Risk}, first introduced in \cite{Ahmadi}, to a suitable space of stochastic processes. Second, we explore the capital allocation problem using this new risk measure in an insurance context and we show that the Euler allocation method is the only method to allocate the requiring capital for this risk measure. 
 
 

The outline of the paper is as follows. In Section 2, we introduce the notion of {\it Cumulative Entropic Value at Risk} (CEVaR${}_{1-\beta}$) as a coherent risk measure on the space of bounded stochastic processes and we explore some of its relevant features. In Section 3, we explore the capital allocation problem and give a theorem which characterizes the capital allocation set for these measures. In fact, we show that for the CEVaR${}_{1-\beta}$ risk measure the Euler allocation method is the only way to allocate the risk capital. Finally, in Section 4, we show some results for CEVaR${}_{1-\beta}$ and provide some examples.

\section{Cumulative Entropic Risk Measures}

Let $(\Omega, \mathcal F, \mathbb P, \bar{\mathcal F})$ be a filtered probability space. We consider the space $\mathcal{R}^p$ of stochastic processes on $[0, T]$ that are c\` adl\` ag, adapted and such that $X^*:=\mathop {\sup}_{[0,T]} |X_t| \in L^p(\Omega,\mathcal F)$, with $1\leq p\leq \infty$. 
 Furthermore, assume that $L^1(\Omega, \mathcal{F},\mathbb P)$ has a countable dense subset. In \cite{Cheridito1} and \cite{Cheridito2} the authors developed the theory of convex risk measures on the space of $\mathcal{R}^p$ ($\rho:\mathcal{R}^p \longrightarrow \mathbb R_+$). 
Notice that, for any $1\leq p\leq \infty$, the space $\mathcal{R}^p$ endowed with the norm $||X||_{\mathcal{R}^p} = ||X^*||_{L^p}$, is a Banach space.
 
\begin{defn}
We define the subspace $\mathcal R_L^p$ containing the processes in $\mathcal R^p$ which have Laplace transform, i.e. $X\in\mathcal R^p$ belongs to $\mathcal R_L^p$ if and only if 
$$ m_t(s)=\mathbb E [\exp(-s\,X_t)]<\infty \;,\qquad s\geq 0 \;,$$
for $t\in[0,T]$. 
\end{defn}

The idea we put forward in this paper is to use a {\it cumulative risk measure} based on the {\it Entropic Value at Risk} that was defined in \cite{Ahmadi}. That is, following \cite {Assa1}, we measure the risk of a random process $X\in \mathcal{R}^p_L $ by defining a {\it cumulative risk measure} $\rho: \mathcal{R}^p_L \longrightarrow \mathbb R_+$ as follows. Let $\rho_{0}$ be a given risk measure on $L^p(\Omega, \mathcal F)$, i.e. $\rho_{0}: L^p(\Omega, \mathcal F)\longrightarrow \mathbb R$, and let $\omega: [0,T] \longrightarrow \mathbb{R}^+$ be a suitable weight function, i.e. $\int_0^T \omega(t) dt = 1$. Then we can define a cumulative risk measure $\rho: \mathcal{R}^p_L \longrightarrow \mathbb R_+$ based on $\rho_0$ as the weighted aggregate risk of a random process $X\in \mathcal R^p_L$. More precisely, 
\begin{equation}\label{cumulative_risk}
\rho(X) := \int_0^T \rho_0(X_t) \omega (t) dt\;.
\end{equation}
Such constructions were proposed and studied in \cite{Assa1}. The features of such measures inherently depend on the choice of base risk measure $\rho_0$. In fact, if the risk measure $\rho_0$ is coherent then $\rho$ in $(\ref{cumulative_risk})$ is coherent as well. 

\begin{thm}\label{Coherent}
Let $\rho_0$ be a coherent risk measure on $L^{\infty}(\Omega, \mathcal F)$. Then the risk measure $\rho: \mathcal{R}^p_L \longrightarrow \mathbb R_+$, given in $(\ref{cumulative_risk})$, is a coherent risk measure on the space $ \mathcal{R}_L^p$.
\end{thm}
\begin{proof}
First we show the positive homogeneity and translation invariance properties of $\rho $. For $\lambda>0$ and $m \in \mathbb{R}$ we have,
\begin{equation*}
\rho(\lambda X + m) = \int_0^T \rho_0(\lambda X_t +m) \omega(t) dt = \lambda ~ \rho(X )- m \int_0^T  \omega(t) dt  \;,
\end{equation*}
which shows the positive homogeneity and translation invariance properties since $\int_0^T  \omega(t) dt = 1$.

As for monotonicity, if $X_t\leq Y_t$ a.s., then $\rho_0(X_t)\geq\rho_0(Y_t)$ for $t\in [0,T]$. Now, since $\omega$ is a positive real valued function, we have $\rho_0(X_t)\omega(t)\geq \rho_0(Y_t)\omega(t)$ for any $t\in [0,T]$ as well. This implies that $\rho(X)\geq \rho(Y)$ which proves the monotonicity property. 

Now using the convexity property of $\rho_0$ and since $\omega$ is a positive function we have, 
$$ \rho_0(X_t+Y_t)\omega(t)\leq  \rho_0(X_t)\omega(t) + \rho_0(Y_t)\omega(t)\;,$$ 
for $t\in [0,T]$. This directly implies the convexity property of $\rho$. i.e., 
\begin{equation*}
\rho(X+Y)\leq  \rho(X) + \rho(Y)\;.
\end{equation*}
\end{proof}

In this paper, we propose to use the {\it Entropic Value at Risk} measure (EVaR${}_{1-\beta}$) as our measure $\rho_0$ in (\ref{cumulative_risk}). This yields an interesting family of risk measures on the space of bounded stochastic processes. Following \cite{Ahmadi} we now give a first definition.

\begin{defn}\label{def:EVAR}
Let $X$ be a random variable in $L^{\infty}(\Omega, \mathcal F)$ having Laplace transform, i.e. 
$$\mathbb E [\exp(-s\,X)]<\infty \;,\qquad s >0\;.$$
Then the {\bf Entropic Value at Risk}, denoted by EVaR${}_{1-\beta}$,  is given by 
\begin{equation}\label{EVAR}
EVaR_{1-\beta}(X) := \mathop{\inf}_{s>0} \frac{\ln \mathbb{E}[\exp(-s\,X)] - \ln\beta}{s}.
\end{equation}
\end{defn}

The following key result for EVaR${}_{1-\beta}$ can be found in \cite{Ahmadi}.
\begin{thm}
The risk measure EVaR${}_{1-\beta}$ from Definition \ref{def:EVAR} is a {\it coherent} risk measure. Moreover, for any $X\in L^{\infty}(\Omega, \mathcal F)$ having Laplace transform, its dual representation has the form
\begin{equation}\label{Robust}
EVaR_{1-\beta}(X) = \mathop{\sup}_{f\in \mathcal{D}} \mathbb E_{\mathbb P}(-fX)\;,
\end{equation}
where $\mathcal{D} = \{ f\in L^1_+(\Omega,\mathcal F) \;|\; \mathbb E_{\mathbb P}[f \ln(f)] \leq - \ln\beta\}$
 and 
\begin{equation}\label{set}
L^1_+(\Omega,\mathcal F) = \{f \in L^1(\Omega,\mathcal F)  \;|\; \mathbb E_{\mathbb P}(f) =1\}.
\end{equation} 
 
\end{thm}

For the proof we refer to \cite{Ahmadi}.\\

If we use the risk measure (\ref{EVAR}) in our general definition of a cumulative risk measure (\ref{cumulative_risk}), we naturally obtain a risk measure on the space  $\mathcal R_L^p$  that would inherit some of the key features of the original risk measure.

We now formally introduce the concept of {\it Cumulative Entropic Value at Risk}, denoted by CEVa$R_{1-\beta}$, on the space  $\mathcal R_L^p$.

\begin{defn}\label{def:CEVAR}
Let $X$ be a stochastic process in $\mathcal R_L^{p}$ and let EVaR${}_{1-\beta}$ be the risk measure in Definition \ref{def:EVAR}. Then, for a given weight function $\omega: [0,T] \longrightarrow \mathbb{R}^+$ (i.e. $\int_0^T \omega(t) dt = 1$), the {\bf Cumulative Entropic Value at Risk}, denoted by CEVaR${}_{1-\beta}$,  is defined by
\begin{equation}\label{CEVAR}
CEVaR_{1-\beta}(X) = \int_0^T EVaR_{1-\beta}(X_t) \omega(t) dt\;.
\end{equation}
\end{defn}

The main advantage of using (\ref{EVAR}) as our based measure is that the resulting cumulative risk measure (\ref{CEVAR}) is tractable enough for a wide family of collective risk models. This comes from the fact that the expectation appearing in (\ref{EVAR}) is merely the Laplace exponent of the random variable $X_t$ (for $t\geq 0$). In collective risk theory, many of the models used for insurance reserves have closed-form Laplace transforms, in particular the so-called L\'evy insurance risk processes. If the aggregate claims process is driven by a spectrally negative L\' evy processes then a cumulative entropic risk measure based on the EVAR${}_{1-\beta}$ is an natural choice to work with in risk management applications.

The risk measure in Definition \ref{def:CEVAR} belongs to the general framework of axiomatic risk measures on the space of stochastic processes developed in \cite{Cheridito1}. We now study some of its properties.
 
\begin{cor}
The risk measure CEVaR${}_{1-\beta}$, given in Definition \ref{def:CEVAR}, is a coherent risk measure on the space $ \mathcal{R}_L^p$.
\end{cor}
\begin{proof}
Since EVaR${}_{1-\beta}$ is of the form (\ref{cumulative_risk}) with a coherent base risk measure $\rho_0$, it follows that EVaR${}_{1-\beta}$ is a coherent risk measure as a special case of Theorem \ref{Coherent}.
\end{proof}

Now, one can notice that in Definition \ref{def:CEVAR} the weight function $\omega$ plays an important role. Different choices of weight functions would result in different cumulative Entropic risk measures. One can naturally think of $\omega$ as a density function that distributes a probability mass over the interval $[0,T]$. Interesting choices would be to use the density function $f_{\tau}$ of a suitable stopping time $\tau \in [0,T]$, like the first passage time or ruin time. This would penalize certain regions of the interval $[0,T]$ according to whether a certain meaningful event is more or less likely to occur over these regions. 


For tractability purposes, in this paper, we use a uniform weight function, i.e. we consider $\omega(t) = \frac1T$. In the remaining of the paper we will be working with the following subfamily of CEVaR, 
 \begin{equation}\label{CEVAR_average}
 CEVaR{}_{1-\beta}(X) = \frac1T\int_0^T EVaR_{1-\beta}( X_t )  dt\; .
 \end{equation} 

Now, the object of our interest in this paper is to apply the {\it CEVaR} in (\ref{CEVAR_average}) within an insurance context where the aggregate claims are modeled by a spectrally positive L\'evy processes. We should then first verify that the class of spectrally positive L\' evy processes can be included in the space $\mathcal{R}^p$ for some $p\geq 1$. This would enable us to use {\it CEVaR} to this class of processes.

\begin{pro}
Let $(X_t)_{0\leq t\leq T}$ be a spectrally positive L\' evy process, then, $(X_t)_{0\leq t\leq T} \in \mathcal{R}^1$.
\end{pro}

\begin{proof}
From the L\' evy-Ito Decomposition we see that every spectrally positive L\' evy process has the following representation (see \cite{Applebaum}),

\begin{equation}
X_t = bt + \sigma B_t+ \int _{0<x< 1} x \tilde{N}(t,dx) + \int _ {x\geq 1} x N (t,dx)\;,
\end{equation}
where $b\in \mathbb{R}$, $\sigma \in \mathbb{R}_+$, $B$ is a standard Brownian motion, $N$ is an independent Poisson random measure and $\tilde{N}$ is the compensated Poisson random measure associated to $N$.

Define $Y_t =  \sigma B_t+ \int _{0<x< 1} x \tilde{N}(t,dx) $ and $Z_t = bt + \int _ {x\geq 1} x N (t,dx)$. It is known, (see \cite{Applebaum}), that the process $(Y_t)_{0\geq t\leq T}$ is a martingale with finite moments. By using Doob's martingale inequality for the process $Y$ we have,
$$||\mathop{\sup}_{0\leq t\leq T}|Y_t|~ ||_p \leq \frac{p}{p-1} ||~|Y_T|~||_p\;, $$
which implies that the process $(Y_t)_{0\leq t\leq T}$ is in $\mathcal{R}^p$ for all $p>1$ . In order to prove the assertion, it is now sufficient to show that the process $(Z_t)_{0\leq t\leq T}$ is in $\mathcal{R}^1$. 

We know that
\begin{eqnarray*}\label{inequality}
\mathbb E[\mathop{\sup}_{0\leq t\leq T} |Z_t|] &\leq& \mathbb E \left[\mathop{\sup}_{0\leq t\leq T} (|b|t +(Z_t-bt))\right]\leq |b|T+ \mathbb E\left[\mathop{\sup}_{0\leq t\leq T} (Z_t-bt)\right]\\
&\leq& (|b|+c)T+ \mathbb E\left[ \mathop{\sup}_{0\leq t\leq T}(-ct+ \int _ {x\geq 1} x N (t,dx))\right] 
\end{eqnarray*}
\noindent for some $c>0$ satisfying the safety loading condition (\ref{Loading}) for the compound Poisson process $\int _ {x\geq 1} x N (t,dx)$. 

Now, it is enough to show that the process $ -ct +\int _ {x\geq 1} x N (t,dx)$ is in $\mathcal{R}^1$.  In order to do that, we use the fact that the process $\int _ {x\geq 1} x N (t,dx)$ is a compound Poisson process with jumps larger than $1$ and therefore the process $( C_t=ct -\int _ {x\geq 1} x N (t,dx))_{0\leq t\leq T}$ can be thought of as the net aggregate process in the classical Cramer-Lundberg model of collective insurance risk theory. We can then define the associated time of ruin $\tau_u = \inf \{t\geq 0 | u +C_t<0 \}$ as well as the associated probability of ruin,
$$ \psi(u) = \mathbb P\left( \mathop{\inf}_{t\geq 0} C_t<-u \right)\;.$$
Notice that the ruin probability is simply the tail of the distribution of the random variable $\mathop{\inf}_{t\geq 0} C_t$ and so we can write for some $\beta \in (0,1)$ ,
$$ \inf \{ u>0 \;|\; \psi(u)\leq \beta\} = VaR_{\beta}\left(\mathop{\inf}_{t\geq 0} C_t\right)\;.$$
Since $\mathop{\inf}_{t\geq 0} C_t\leq \mathop{\inf}_{0\leq t\leq T} C_t$, it can be seen by the definition of $VaR$ that $VaR_{\beta}\left(\mathop{\inf}_{0\leq t\leq T} C_t\right)\leq VaR_{\beta}\left(\mathop{\inf}_{t\geq 0} C_t\right)$.
It is well-known \cite{Bowers}, that $\psi(u)\leq e^{-Ru}$, where $R$ is the smallest positive root of the {\bf Lundberg's fundamental equation}, i.e.,  $ \lambda + cr=\lambda M_{\int _ {x\geq 1} x N (t,dx)}(r)$  where $\lambda$ is the intensity rate for the compound poisson process  $\int _ {x\geq 1} x N (t,dx)$ and $M_{\int _ {x\geq 1} x N (t,dx)}(r)$ is the moment generating function for $\int _ {x\geq 1} x N (t,dx)$.

This implies that for $T>0$,
\begin{equation}\label{inequality2}
VaR_{\beta}\left(\mathop{\inf}_{0\leq t\leq T} C_t\right) \leq  VaR_{\beta}\left(\mathop{\inf}_{t\geq 0} C_t\right)\leq \frac{- \ln\beta}{R}\;,
\end{equation}
which in turn implies that,
\begin{equation*}
\mathbb E \left[\mathop{\sup}_{0\leq t\leq T} (-C_t)\right] = - \mathbb E\left[ \mathop{\inf}_{0\leq t\leq T} C_t\right] = \int_0^1 VaR_{\beta}\left(\mathop{\inf}_{0\leq t\leq T} C_t\right) d\beta \;,
\end{equation*}
where the last inequality comes from the integral representation of the expectation in terms of its quantiles. Using (\ref{inequality2}), we can finally write,
\begin{equation*}
\mathbb E \left[\mathop{\sup}_{0\leq t\leq T} (-C_t)\right] \leq  -\frac1R \int_0^1 \ln\beta d\beta < \infty\;,
 \end{equation*}
which implies that the process $Z_t$ is in $\mathcal{R}^1$. This completes the proof.
\end{proof}


\subsection{Examples}\label{sec:examples}
The {\it Cumulative Entropic Risk Measure} introduced in Definition \ref{def:CEVAR} has the advantage of being tractable enough for a large family of processes which have Laplace transform and that can be used as models for the net-loss process in (\ref{risk}). Here we discuss a few examples and compute expressions for the CEVaR in (\ref{CEVAR_average}) for some L\'evy insurance risk models. 

\subsubsection{Brownian Motion with Drift}
Let $Y_t = \mu t + \sigma W_t$ be a Brownian motion with drift parameter $\mu$ and scale parameter $\sigma$ for $\mu\in \mathbb{R}, \sigma >0$. Such a process are used in collective risk theory as the net-loss process in (\ref{risk}) for an approximation to the classical Cramer-Lundberg model (\cite{Grandell}). The Laplace transform of $Y_t$ is 
\begin{equation*}
\mathbb{E}(e^{-sY_t}) = e^{-\mu ts + \frac12 \sigma^2 s^2 t}\;.
\end{equation*}
By direct substitution in $(\ref{EVAR})$ and differentiation with respect to $s$ we have, for $t\in[0,T]$,
\begin{equation}\label{EVAR_BM}
EVaR_{1-\beta}(Y_t) =  -\mu t+ \sigma \sqrt{-2t\ln\beta }\;. 
\end{equation}
Direct substitution and integration of (\ref{EVAR_BM}) into (\ref{CEVAR_average}) results in 
$$CEVaR_{1-\beta}(Y)= \frac{-\mu T}{2} + \frac23 \sigma \sqrt{-2T\ln\beta}\;.$$
We observe that this risk measure is an increasing linear function of $\sigma$ and a decreasing linear function of the premium rate $\mu$. This is intuitively natural since a larger premium decreases the risk exposure whereas a large volatility in the claims severities produces an increase in the risk exposure.


\subsubsection{$\alpha$-stable Subordinator}
Let $Y_t= \mu t + X_t$ be an $\alpha$-stable subordinator with drift with Laplace exponent $\phi(u) := -\frac{1}{t} \ln \mathbb{E}(e^{-uY_t}) = \mu +u^{\alpha}$ for $0<\alpha<1$. This model has been used as net-loss process in \cite{Furrer}. By applying the same straight-forward procedure as in the previous example we have, for $t\in[0,T]$,
\begin{equation*}
EVaR_{1-\beta}(Y_t) =  -\alpha t (\frac{\ln\beta}{\alpha t -1})^{\frac{\alpha -1}{\alpha}}-\mu t.
\end{equation*}
Direct integration over $[0,T]$ yields, 
\begin{equation*}
CEVaR_{1-\beta}(Y) = \alpha \ln\beta \left[ \frac{\alpha T -1}{\alpha T + T}-1\right]\left(\frac{\alpha T - 1}{\ln\beta}\right)^{\frac1{\alpha}}- \frac {\mu T}{2}.
\end{equation*}

\subsubsection{Gamma Subordinator}

Let $Y_t = \mu t +X_t $ be a gamma
process with parameters $a, b > 0$ with drift.  
This process has Laplace exponent $\phi(u):=-\frac{1}{t} \ln  \mathbb{E}(e^{-sY_t}) = -\mu t u -ta \ln( 1+ u/b)$.  In this case, EVaR$_{1-\beta}(Y_t)$ is the solution of the following equation, 
$$ -\frac{tau}{b+u} + at \ln (1+\frac ub) + \ln\beta = 0\;,\qquad u \geq 0\;. $$
The above equation is obtained by applying the Laplace exponent in the definition EVaR$_{1-\beta}$ and by straight-forward differentiation with respect to $u$. 

Unlike the previous examples, there is no close-form expression for the solution of this equation. But once EVa$R_{1-\beta}$ is obtained numerically, we can calculate CEVaR$_{1-\beta}(Y)$ by direct integration over $[0,T]$.

\section{Capital Allocation}

We now study the problem of capital allocation in an insurance context with the coherent risk measure {\it CEVaR} that we introduced in the previous section. Discussing the problem of capital allocation for {\it CEVaR}, which is a risk measure defined on $\mathcal R^p_L$, must start with an analysis of this problem for {\it EVaR}, which is a risk measure on a subspace of $L^{\infty}(\Omega,\mathcal F)$. 

Finding the capital allocation for a risk measure on the space of stochastic processes typically requires knowledge of its robust representation and its sub gradient set (see \cite{Assa2} for a detailed account on this problem). This robust representation is typically a hard problem in the space $\mathcal{R}^p_L$ that normally requires functional analysis tools. In the case of {\it EVaR} we propose to handle the issue of finding capital allocation for CEVaR by finding the capital allocation for EVaR and use the linear relation between EVaR and CEVaR to get the capital allocation for CEVaR.

We now give some definitions that will be needed throughout this section. 

\begin{defn}
Let $\rho$ be a coherent risk measure defined on $ L^{\infty}(\Omega,\mathcal F)$. Now let $\mathcal{D} \subset L^1_+$ be the largest set for which the following robust representation holds true for $\rho$,
\begin{equation}\label{Dertermining}
\rho(X) =\mathop {sup}_{f\in \mathcal{D}\subset L^1_+} \mathbb E_{\mathbb P}(-fX) ~~~~ \forall X\in L^{\infty}(\Omega,\mathcal F)\;,
\end{equation}
where $L^1_+$ is the set defined in $(\ref{set})$. The set $\mathcal{D}$ is called the determining set of $\rho$ (see \cite{Follmer2}).
\end{defn}

The following definition is taken from \cite{Cherny}.


\begin{defn}
Let $\rho$ be a coherent risk measure defined on $ L^{\infty}(\Omega,\mathcal F)$ with determining set $\mathcal{D} \subset L^1_+$.  Let $X \in L^{\infty}(\Omega,\mathcal F)$. A function $ f \in \mathcal{D} $ is called an {\it extreme function} for $X$ if $\rho(X)= \mathbb E_{\mathbb P}(fX) \in (-\infty, \infty)$. The set of extreme functions will be denoted by $\chi_{\mathcal{D}} (X)$. 
\end{defn}

The following result is taken from \cite{Cherny} and gives conditions for the set of extreme functions defined above to be non-empty.

\begin{pro}
Let $\mathcal{D} \subset L^1_+$ be the determining set of a given coherent risk measure $\rho$ on $ L^{\infty}(\Omega,\mathcal F)$. Now consider the following set, 
\begin{equation}\label{conjunto}
L^1(\mathcal{D}) := \{ X\in L^{\infty}(\Omega,\mathcal F)  \;|\; \mathop{\lim}_{n \longrightarrow \infty} \mathop{\sup}_{f\in \mathcal{D}} \mathbb E_{\mathbb P}[f\,|X|\, \mathbb I_{\{|X|> n\}}] = 0\}.
\end{equation}
If the determining set $\mathcal{D}$ is weakly compact and $X \in L^1(\mathcal{D})$, then the set of extreme functions for $X$ is not empty, i.e.$ \chi_{\mathcal{D}}(X)\neq \emptyset$.
\end{pro}

Now, we turn our attention to the concept of {\bf capital allocation}. Consider a vector of risks $X = ( X^1, \dots ,X^d)$, such that $X^i \in L^{\infty}(\Omega, \mathcal F)$ for $i=1,\dots,d$, are random variables representing the cash flow or risk exposure 
of a portfolio consisting of $d$ risky positions or departments. In this paper, these will be net-loss positions of an insurance policy contract at a given time. 

Given a coherent risk measure $\rho$ on $L^{\infty}(\Omega,\mathcal F)$, we now look at the problem of how to allocate the total risk of the portfolio $\rho \left( X^1 + \dots + X^d \right)$ among the different departments such that the individual risk of each one of them is properly measured. 

The following formal definition of capital allocation was proposed by \cite{Delbaen3} and \cite{Fischer} and it is the one we set out to study in this paper. In fact, the following gives a mathematical definition of capital allocation for a coherent risk measure. 

\begin{defn}\label{Capital}
Consider a coherent risk measure $\rho$ on $L^{\infty}(\Omega, \mathcal F)$ and a vector of risks $X=(X^1, \dots ,X^d)$ such that $X^i \in L^{\infty}(\Omega, \mathcal F)$ for $i=1,\dots,d$. A fair capital allocation for $X$ is a vector $(K^1, . . . , K^d ) \in \mathbb{R}^d$ such that
\begin{enumerate}
\item \begin{equation*}
\sum_{i=1}^d K^i = \rho \left(\sum_{i=1}^d X^i \right)\;,  \\
\end{equation*}
\item  
\begin{equation*}
\sum_{i=1}^d  h^i K^i \leq  \rho\left(\sum_{i=1}^d h^i  X^i \right)\;,\qquad \forall \, h = (h^1,\dots, h^d) \in \mathbb{R}^d_+\;.
\end{equation*}
\end{enumerate}
\end{defn}
The first condition is called the {\it full allocation} property and it simply states the fact that the total risk of the whole portfolio should be the aggregated risks of each department. The second condition is called the {\it linear diversification property} of capital allocation. In fact, this condition has a one to one correspondence with the positive homogeneity and subadditivity properties of a coherent risk measure $\rho$ (see \cite{Kalkbrener}). Since we work in this paper with a coherent risk measure it is somehow natural to adopt this definition of capital allocation.

The following is an interesting result characterizing the set of possible such capital allocations and it is adapted from \cite{Cherny}. 
 
\begin{thm}\label{capital1}
Let $\mathcal D \subset L^1_+$ be the determining set of a given coherent risk measure $\rho$ on $L^{\infty}(\Omega,\mathcal F)$ and let $X=(X^1, \dots ,X^d)$ be a vector such that $X^i \in L^{\infty}(\Omega, \mathcal F)$ for $i=1,\dots,d$. Consider the following set 
\begin{equation}\label{setG}
G = \overline{\{\left( \mathbb E_{\mathbb P}(-f\,X^1), \dots \, , \mathbb E_{\mathbb P}(-f\,X^d) \right) \;|\; f \in \mathcal{D}\}} \subset \mathbb{R}^d \;.
\end{equation} 
The set $U\subset \mathbb R^d$ of capital allocations for $X=(X^1, \dots ,X^d)$, satisfying Definition \ref{Capital}, is convex and bounded and it has the form
\begin{equation}\label{argmax}
U = \mathop{argmax}_{x \in G} <e,x>\;,
\end{equation}
where $<\cdot,\cdot> $ is the inner product in $\mathbb{R}^d$, $e = (1, \dots, 1)$ and argmax is the set of points of $G$ for which $<e,x>$ attains its maximum value. 

If moreover, $X^1, \dots ,X^d \in L^1(\mathcal{D})$ and $ \mathcal{D}$ is weakly compact, then $U$ can be identified to be
\begin{equation}
U = \left\{ \left( \mathbb E_{\mathbb P}(-f\,X^1), \dots\, , \mathbb E_{\mathbb P}(-f\,X^d) \right) \;  |\; f \in \chi_{\mathcal{D}} \left( \sum_{i=1}^d X^i \right) \right \}. 
\end{equation}
\end{thm}
\begin{proof}
In \cite{Cherny}, the author provides a proof of the theorem for coherent utility functions. The result follows by noticing that, for a given coherent risk measure $\rho$, if we set $\rho^*(X) := -\rho(-X)$ we obtain a coherent utility function and the result in \cite{Cherny} holds. So, from $\rho(X) = -\rho^*(-X)$  the results for the statement of our theorem holds.
\end{proof}

The set $G\subset \mathbb R^d$ in Theorem \ref{capital1} is called the {\it generator} for $X$ and $\rho$ (see \cite{Cherny}). The following corollary gives a condition on $G$ for the uniqueness of the capital allocation.

\begin{cor}\label{uniqueness}
Under the conditions of Theorem \ref{capital1}. If moreover, $G \subset \mathbb R^d$ is strictly convex (i.e. its interior is non-empty and its border contains no interval), then there is a unique capital allocation satisfying Definition \ref{Capital}.
\end{cor}
\begin{proof}
See \cite{Cherny} for a proof in terms of coherent utility functions. 
\end{proof}

Corollary \ref{uniqueness} gives us sufficient conditions for this capital allocation to be unique. The following result characterizes such unique capital allocation by giving a representation for each one of its components.   

\begin{thm}\label{Riskcont}
Let $\mathcal D \subset L^1_+$ be the determining set of a given coherent risk measure $\rho$ on $L^{\infty}(\Omega,\mathcal F)$ and $L^1(\mathcal{D})$ be the associated set defined in (\ref{conjunto}). Moreover, let $X=(X^1, \dots ,X^d)$ be a vector such that $X^i \in L^{\infty}(\Omega, \mathcal F)$ for $i=1,\dots,d$. If $ \mathcal{D}$ is weakly compact, $X^1, \dots ,X^d \in L^1(\mathcal{D})$ and  $\chi_{\mathcal{D}}(\sum_{i=1}^d X^i ))$ is a singleton then, for $1\leq i \leq d$,
\begin{equation}
\mathbb E_{\mathbb P}(-fX^i) = \mathop {\lim}_{\epsilon \downarrow 0} \frac{\rho(\sum_{j=1}^d X^j  + \epsilon X^i) - \rho(\sum_{j=1}^d X^j)}{\epsilon} \;.
\end{equation}
\end{thm}
\begin{proof}
Since $\chi_{\mathcal{D}}(\sum_{i=1}^d X^i )$ is a singleton and $ \mathcal{D}$ is weakly compact then, from Theorem \ref{capital1}, there exists an unique function $f\in \chi_{\mathcal{D}}(\sum_{i=1}^d X^i )$ and as a consequence the set of capital allocations is $U = \{ (\mathbb E_{\mathbb P}(-f X^1), \dots, \mathbb E_{\mathbb P}(-f X^d))\}$. This means that each component of the capital allocation $\mathbb E_{\mathbb P}(-f X^i)$ is simply the {\it risk contribution} of $X^i$ to $\sum_{i=1}^d X^i $ (see \cite{Cherny}). If we use the standard notation for the $i^{th}$ risk contribution $\tilde{\rho}(X^i| \sum_{i=1}^d X^i )$ this means, 
\begin{equation*}
\tilde{\rho}(X^i| \sum_{i=1}^d X^i ) = \mathop{\sup} _{f\in \chi_{\mathcal{D}}(\sum_{i=1}^d X^i )} \mathbb{E}_{\mathbb P}(-f X^i)\;.
\end{equation*}

It is then sufficient to show that, 
\begin{equation}\label{risk_contribution}
\tilde{\rho}(X^i| \sum_{i=1}^d X^i ) = \mathop {\lim}_{\epsilon \downarrow 0} \frac{\rho(\sum_{j=1}^d X^j  + \epsilon X^i) - \rho(\sum_{j=1}^d X^j)}{\epsilon}\;,~~~~ for~ 1\leq i \leq d\;.
\end{equation}
The result in (\ref{risk_contribution}) was shown in \cite{Cherny} for coherent utility functions. Since $\rho^*(X) := -\rho(-X)$ is a coherent utility function for any coherent risk measure $\rho$, it follows directly that (\ref{risk_contribution}) also holds for coherent risk measures. 
\end{proof}

\subsection{CEVaR and the Capital Allocation Problem}

Our main goal in this paper is to apply cumulative entropic risk measure in a capital allocation problem. In the previous section, we discussed key notions of the capital allocation problem for a risk measure on $L^{\infty}(\Omega,\mathcal F)$. In this section, we apply these results in order to give an answer to the problem of capital allocation for {\it CEVaR} which is a risk measure on $\mathcal R_L^p$. Notice that this is a somewhat more complicated problem since there is a dynamic component to this problem. Here, this is overcome by the cumulative property of CEVa$R_{{}_{1-\beta}}$. We start by extending Definition \ref{Capital} to the more general notion of capital allocation with respect to a coherent risk measure on the space $\mathcal R_L^p$. The following definition is taken from \cite{Billera}.

\begin{defn}\label{capital for process}
Let $\left(X^1_t,\dots, X^d_t\right)_{t\in[0,T]}$ be $d$ random processes in $\mathcal R_L^p$ representing $d$ financial positions or departments. Moreover, consider a coherent risk measure $\rho: \mathcal{R}^p_L \longrightarrow \mathbb R_+$ defined on the space $\mathcal R^p_L$. A fair capital allocation for $\left(X^1_t,\dots, X^d_t\right)_{t\in[0,T]}$ with respect to $\rho$ is a vector $(L^1, . . . , L^d ) \in \mathbb{R}^d$ such that,
\begin{enumerate}
\item \begin{equation*}
\sum_{i=1}^d L^i = \rho \left(\sum_{i=1}^d X^i \right)\;,  \\
\end{equation*}
\item  
\begin{equation*}
\sum_{i=1}^d  h^i L^i \leq  \rho\left(\sum_{i=1}^d h^i  X^i \right)\;,\qquad \forall \, h = (h^1,\dots, h^d) \in \mathbb{R}^d_+\;,
\end{equation*}
\end{enumerate}
\end{defn}
 where $\sum_{i=1}^d X^i$ denotes the process $\left( \sum_{i=1}^d X_t^i\right)_{t\in[0,T]}$.
 
In this section, we show how a capital allocation satisfying Definition \ref{capital for process} can be obtained when using CEVaR as risk measure. 

We first need to show that the border of the set $\mathcal{D}$ in the robust representation (\ref{Dertermining}) for EVa$R_{1-\beta}$ is not a convex set. This leads to the fact that the set $G$ in $(\ref{setG})$ is not a convex set too. In fact, this immediately implies that the {\it Euler} allocation (see \cite{Tasche3})  is the only possible allocation method for {\it EVaR} as well as for {\it CEVaR}.  

\begin{thm}\label{convex}
Let $D$ be the determining set in the robust representation $(\ref{Robust})$ for EVa$R_{1-\beta}$. Then the set $\partial \mathcal{D} =\{ f\in L^{\infty}_+: \mathbb E_{\mathbb P}(f \ln(f)) = - \ln\beta\}$ is not a convex set.
\end{thm} 

\begin{proof}
It is sufficient to show that for any $\lambda \in [0,1]$ and any two functions $f$ and $g$ in $\partial \mathcal{D}$, the function $\lambda f + (1- \lambda) g$ is not in $\partial \mathcal{D}$. Define the function $H$ on the space of positive real line taking real values as follows, 
$$ H(x) := x \ln x,$$
 for all $x \in \mathbb{R}^+.$

It is clear that the function $H$ is strictly convex on the positive real line. Since, $H'(x)=  \ln x + 1$ and $H''(x) = \frac1x>0$ for all $x \in \mathbb{R}_+.$ Now again we define a new function $K$ on $[0,1]$  with its values in $\mathbb{R}$ by using the composition function $ H(\lambda f + (1- \lambda) g)$ as follows, 
$$K(\lambda) = E_{\mathbb P}(H(\lambda f + (1- \lambda) g)),$$ 
for the fixed functions $f$ and $g $ in $\partial \mathcal{D}$. Notice that we use a slight abuse of notation, here $H(\lambda f + (1- \lambda) g)$ is to be understood point-wise. That is, for $x\in \mathbb R$, the function $ H(\lambda f + (1- \lambda) g) \longrightarrow H(\lambda f(x) + (1- \lambda) g(x))$. 

If we take the first and second derivatives for the function $K$, we see that this function is strictly convex too. $K'(\lambda) = E_{\mathbb P}( (f-g) (H'(\lambda f + (1- \lambda) g))$ and 
$K''(\lambda) =  E_{\mathbb P}((f-g)^2 (H''(\lambda f + (1- \lambda) g)) = \frac{(f-g)^2}{\lambda f + (1- \lambda) g} >0$. Now, considering  $K(0) = E_{\mathbb P}(H(f))$ and $K(1) = E_{\mathbb P}(H(g))$  along with the strictly convexity of the function $K$, we come up with the inequality 

$$ K(\lambda) = E_{\mathbb P}(H(\lambda f + (1- \lambda) g)) < - \ln \beta~~~~ \forall \lambda \in (0,1).$$ 
This proves our assertion.
\end{proof}

The uniqueness of the so-called Euler allocation method is stated in the following result.
\begin{pro}\label{capital}
Let $(X^1, \dots ,X^d)$ be a vector such that each $X^i \in L^{\infty}(\Omega, \mathcal F)$, for $i=1,\dots,d$, represents the cash-flow or risk exposure from one risk position or department. We denote by $X = \sum_{i=1}^n X^i$ the portfolio-wide cash-flow produced over a given time-period. Furthermore, for a given risk measure $\rho$ on $L^{\infty}(\Omega,\mathcal F)$, define the function $f_{\rho}(u_1,\dots, u_n) = \rho( \sum_{i=1}^n u_i X^i)$. If $\rho$ is EVa$R_{1-\beta}$ as defined in $(\ref{def:EVAR})$, then the capital allocated to each department that satisfies Definition \ref{Capital} is determined uniquely by,
\begin{equation}\label{Euler}
K^i = \frac{d\rho}{dh}(X +hX^i)|_{h=0} =\frac{\partial}{\partial u_i}f_{\rho}(1,\dots,1) \qquad 1\leq i \leq n\;.
\end{equation}
\end{pro}

\begin{proof}
As a result of Theorem \ref{convex} we see that the associated set $G$ in $(\ref{setG})$  is strictly convex. i.e, the capital allocation for the vector $X = (X^1,\dots, X^d)$ is unique. Since, the risk measure EVaR$_{1-\beta}$ is positive homogeneous, i.e., for all $\lambda >0$ we have 
EVaR$_{1-\beta}(\lambda X) =$ $\lambda$ EVaR$_{1-\beta}(X)$, we deduce that the function $f_{\rho}$ above is a homogeneous function. So, by Euler's theorem on homogeneous functions we have 
\begin{equation*}
f_{\rho}(u_1,\dots,u_n) = \sum_{i=1}^n u_i \frac{\partial}{\partial u_i} f_{\rho}(u_1,\dots,u_n)\;.
\end{equation*}
Now, by applying Theorem \ref{Riskcont} and evaluating $\frac{\partial}{\partial u_i}f_{\rho}$ at $(1,\dots, 1)$, we deduce that the capital allocated to each department is given by $(\ref{Euler})$. This proves the theorem.
\end{proof}
Now, we are going to characterize the capital allocation satisfying Definition \ref{capital for process} with respect to CEVa$R_{1-\beta}$. Notice that this seems to be a more complicated problem since CEVa$R_{{}_{1-\beta}}$ is a risk measure defined on the space of stochastic processes $\mathcal R_L^p$. However, this is possible thanks to the cumulative property of CEVa$R_{{}_{1-\beta}}$. 

\begin{thm}\label{main_co}
Let $(X_t^1, \dots ,X_t^d)_{0\leq t\leq T}$ be a vector such that each $\left( X^i_t\right)_{0\leq t\leq T} \in  \mathcal R_L^{p}$ (for $i=1,\dots,d$) represents the cash-flow or risk exposure from one risk position or department at time $t\in[0,T]$. Then, the capital allocation satisfying Definition \ref{capital for process} over the period $[0,T]$, with respect to CEVa$R_{1-\beta}$, is determined uniquely for $i=1,\dots,d$ by,
$$ L^i = \frac 1T \int_0^T K_t^i dt\;,$$
where $K_t^i$ is given by $(\ref{Euler})$.
\end{thm}
\begin{proof}
By replacing the representation for CEVa$R_{1-\beta}$ given in equation $(\ref{CEVAR_average})$ into Definition \ref{capital for process}, we can see that the capital allocation reduces to the one for EVa$R_{1-\beta}$ given in Definition \ref{Capital}. Equation $(\ref{Euler})$ immediately implies the result.
\end{proof}
Theorem \ref{main_co}, gives us a solution to the problem of capital allocation for stochastic processes over a finite time period $[0,T]$. Interesting enough, unlike other solutions to this problem, this capital allocation can be readily computed for a large family of processes. Now, we turn our attention to an application of our results. 

\subsection{Capital Allocation for Insurance L\'evy Risk Processes}


We now apply Theorem \ref{main_co} to give an answer to the capital allocation problem for an insurance risk process. We consider here an insurance company consisting of $n$ departments. For each department, we let $R_t^i$ be a risk reserve process of the form (\ref{risk}). In other words, $R_t^i=x^i-Y_t^i$ where $Y_t^i=X_t^i-c^it$ denotes the net-loss claim process related to the $i$th department. We recall that $x^i$ is the initial reserve, $c^i$ is the loaded premium and $X_t^i$ is a model for the aggregate claims while the index $i$ refers to one of the $n$ departments. In order to allow for a more rich description of an insurance company, we think of the aggregate claims process $X_t^i$ as the aggregate amount paid out by the department $i$ which is composed of fractions of $m$ independent classes of claims. That is, let $W_t^1, \dots, W_t^m $ be $m$ independent spectrally positive L\'evy process modeling aggregate claims of $m$ different types. One can think for instance of claims associated with car accidents, home damage, medical insurance, etc. Then, the aggregate claims $X_t^i$ paid out by the $i$th department would be a linear combination of some of these $W_t^m$ claims processes. For example, consider aggregate claims produced by a car insurance contract. We suppose that one department will pay out property damage coverage (a fraction of the aggregate claims from the contract) while another department will pay out third-party liability costs (another fraction of the aggregate claims from the contract). 
  
Mathematically, we let $W_t^1, \dots, W_t^m $ be $m$ independent spectrally positive L\'evy processes for $j=1,\dots,m$. Now, we let each $X_t^i$ to be a linear combination of some, or all, of the $W_t^1, \dots, W_t^m $, i.e.
\begin{equation}\label{matrix}
X = \left(\begin{array}{c}
  X_t^1 \\
  X_t^2\\
  \vdots\\
  X_t^n
\end{array}\right) =\left(
\begin{array}{ccc}
  a_{11} & \dots& a_{1m}\\
  a_{21} & \dots& a_{2m} \\
  \vdots & \dots& \vdots\\
  a_{n1} & \dots & a_{nm}
\end{array}\right) \left(\begin{array}{c}
  W_t^1 \\
  W_t^2\\
  \vdots\\
  W_t^m
\end{array}\right) \;, 
\end{equation}
where $a_{ij}$'s are non-negative real numbers for $1\leq i\leq n$ and $1\leq j\leq m$. 

We point out that we chose this structure because it admits a neat solution for the capital allocation problem through Theorem \ref{main_co}. One can always fall back on the more simple case where each department pays out one, and only one, type of claims as oppose to paying fractions of different types of claims. This would correspond to having $n=m$ and a diagonal matrix in (\ref{matrix}) with all elements in the diagonal equal to one yielding $X_t^i=W_t^i$ for all $i$. We also point out that this construction endows the processes $R^i$'s with a dependence structure through the aggregate claims $X^i$'s. The next result is one of the main contribution of our paper.

\begin{thm}\label{thm20}
Consider $n$ risk processes such that $\left( R_t^i \right)_{0\leq t \leq T} \in \mathcal R^p_L$, for $i=1,\dots,n$. Now, let such $R_t^i=x^i-Y_t^i$ where $Y_t^i=X_t^i-c^it$ denotes the net-loss claim process related to the $i$th department. Moreover, let the aggregate risk processes $X^i_t$ be those defined in $(\ref{matrix})$. Then the capital allocation that satisfies Definition \ref{capital for process} over the time period $[0,T]$, for each net-loss process $Y_t^i$ and with respect to the risk measure CEVa$R_{1-\beta}$ is, 
\begin{equation}\label{Main1}
L^i = \frac 1T \int_0^T K_t^i dt + c^i \frac T2, 
\end{equation}
\noindent where 
\begin{equation}\label{Main2}
K_t^i = -t \sum_{j=1}^m a_{ij} \phi_j'(s^*\sum_{k=1}^n a_{kj}) \;,\qquad t\in[0,T]\;,
\end{equation}
and $\mathbb E(e^{-sW_1^j}) = e^{-\phi_j(s)}$ for $s\geq 0$, $\phi'_j(0)\leq0$,  $1\leq i\leq n$ and $1\leq j\leq m$.
\end{thm}

\begin{proof}
First we want to find the capital allocation with respect to the risk measure EVa$R_{1-\beta}$ before applying Theorem \ref{main_co}. For any coherent risk measure $\rho$ defined on $L^{\infty}(\Omega, \mathcal F)$, we have, by the cash-invariant property, that, for each $t\in [0,T]$,
$$ \rho(\sum_{i=1}^n Y_t^i) = \rho(\sum_{i=1}^n X_t^i) +\sum_{i=1}^{n}c^i\,t\;.$$
That is, in order to find the capital allocation (at $t\in[0,T]$) in this setting with respect to a coherent risk measure (in particular for EVa$R_{1-\beta}$), we just need to find the capital allocation for each claim process $X_t^i$.

For a given coherent risk measure $\rho$ on $L^{\infty}(\Omega, \mathcal F)$, let us define the function $f_{\rho}(u_1,\dots, u_n) := \rho( \sum_{i=1}^n u_i X_t^i)$. Taking into account the structure of the processes $X_t^1, \dots, X_t^n $, we can write, for $t\in[0,T]$,
\begin{eqnarray*}
f_{\rho}(u_1,\dots, u_n) &=& \rho( \sum_{i=1}^n u_i X_t^i)\\
&=& \rho\left((\sum_{j=1}^m u_1 a_{1j} W_t^j) +(\sum_{j=1}^m u_2 a_{2j} W_t^j) + \dots + (\sum_{j=1}^m u_n a_{nj} W_t^j) \right)\\
&=&  \rho\left((\sum_{i=1}^n u_i a_{i1}) W_t^1 +(\sum_{i=1}^n u_i a_{i2}) W_t^2 + \dots + (\sum_{i=1}^n u_i a_{in}) W_t^m\right). 
\end{eqnarray*}

If we let  
\begin{equation}\label{d_j}
d_j = \sum_{k=1}^n u_k a_{kj},
\end{equation}
 we can write a more compact form,\begin{equation}\label{f_compact}
f_{\rho}(u_1,u_2,\dots, u_n) = \rho( d_1W_t^1+ d_2 W_t^2 + \dots + d_mW_t^m)\;.
\end{equation}

By using the independence of principal factors $W^i$, we have, for $t\in[0,T]$,
\begin{eqnarray*}
\ln \left( \mathbb E(e^{-s(d_1W_t^1+ d_2 W_t^2 + \dots + d_mW_t^m)})\right) &=& \ln \left( \Pi_{j=1}^m \mathbb E(e^{-sd_j W_t^j})\right)\\
\sum_{j=1}^m \ln \left(\mathbb E(e^{-sd_j W_t^j})\right)&=& -t\sum_{j=1}^m \phi_j(sd_j) \;,
\end{eqnarray*}
where the last equality comes from $\mathbb E(e^{-sW_t^j}) = e^{-t\phi_j(s)}$. 

If we specialize the above equations to the case of EVaR, then equation (\ref{f_compact}) becomes, for $t\in[0,T]$,

\begin{equation}\label{relation2}
f_{EVaR_{1-\beta}}(u_1,u_2,\dots, u_n) = EVaR_{1-\beta}( d_1W_t^1+ d_2 W_t^2 + \dots + d_mW_t^m) = \mathop{\inf}_{s\geq 0} \frac{-t\sum_{j=1}^m \phi_j(sd_j) - \ln\beta}{s}\;.
\end{equation}
Now, consider the right-hand side of equation (\ref{relation2}). By taking derivatives with respect to $s$ 
we have, for $t\in[0,T]$,

\begin{equation}\label{relation3}
\frac{\partial}{\partial s}\left(\frac{-t\sum_{j=1}^m \phi_j(sd_j) - \ln\beta}{s} \right) = \frac{-st\sum_{j=1}^m d_j\phi_j'(sd_j) +t\sum_{j=1}^m \phi_j(sd_j) + \ln\beta}{s^2}.
\end{equation}

By setting equation (\ref{relation3}) equal to zero, we can find the value $s^*(t,u^1,\dots,u^n)$ that minimizes the right-hand side in (\ref{relation2}). As indicated by the notation, this minimum value $s^*(t,u^1,\dots,u^n)$ is a function of $t$ and $u_i$ for $1\leq i\leq n$ but in the following we use the more simple notation $s^*$ for this value. Notice that the value $s^*$ is  in fact the infimum too. Based on convexity property of Laplace transform for one-sided L\' evy processes and the condition $\phi'_j(0)\leq 0$, the infimum in $(\ref{relation2})$ should be reached at some point we denote $s^*$ (see \cite{Kyprianou}). 

According to Proposition \ref{capital}, the {\it Euler} allocation is the only possible allocation method for EVa$R_{1-\beta}$. So, in order to find the capital allocation, it is sufficient to find the derivative of the right-hand side of equation (\ref{relation2}) with respect to the variable $u_i$ and evaluate it at the point $u = (1,1,\dots, 1)$. 
Straight-forward differentiation yields, for $i=1,\dots,n$ and $t\in[0,T]$,
\begin{equation}\label{relation4}
\frac{\partial }{\partial u_i} f_{EVaR_{1-\beta}}(u_1, u_2,\dots, u_n)= \frac{-s^* t\sum_{j=1}^m( s^*_id_j + a_{ij} s^*)\phi_j'(s^* d_j) +ts^*_i \sum_{j=1}^m \phi_j(s^*d_j) + s^*_i\ln\beta}{s^{*2}}\;,
\end{equation}
where we use the notation $s^*_i = \frac{\partial s^*}{\partial u_i}$. 

Since $s^*$ is the solution of setting equation (\ref{relation3}) equal to zero, we can simplify $(\ref{relation4})$ as follows, for $i=1,\dots,n$,
\begin{equation}\label{relation5}
\frac{\partial }{\partial u_i} f_{EVaR_{1-\beta}}(u_1, u_2,\dots, u_n)= -t\sum_{j=1}^m a_{ij}\phi_j'(s^*d_j)\;.
\end{equation}

Evaluating equation (\ref{relation5}) at the point $u = (1,1,\dots, 1)$ yields the allocated capital associated to the $i^{th}$ department at time $t\in [0,T]$. Namely, for $i=1,\dots,n$,
\begin{equation}\label{relation6}
K_t^i = \frac{\partial }{\partial u_i} f_{EVaR_{1-\beta}}(u_1, u_2,\dots, u_n)= -t\sum_{j=1}^m a_{ij}\phi_j'(s^*\sum_{k=1}^n a_{kj})\;.
\end{equation}

Using Theorem \ref{main_co} and integrating $K_t^i$ in $(\ref{relation6})$ yields the allocated capital satisfying Definition \ref{capital for process} with respect to the risk measure CEVa$R_{\beta}$. Thus, the allocated capital to $i^{th}$ department over the period $[0,T]$ with respect to CEVa$R_{1-\beta}$ is, 
$$L^i = \frac 1T \int_0^T K_t^i dt + c^i \frac T2\;.$$
This completes the proof.
\end{proof}

\section{Examples}

%

In this section, we are interested in examining Theorem \ref{thm20} for some examples in order to illustrate how this capital allocation can be computed. We present capital allocations for the examples already discussed in Section \ref{sec:examples}. 

As we will see, there are some cases for which we can obtain an explicit expression for the capital allocation. In others, such an explicit form is not available but a solution can still be obtained by standard numerical methods. The difficulty lies in solving the equation (\ref{relation3}) when is set to be equal to zero. 

\subsection{Brownian Motion with Scale Parameter}

Consider the general set-up defined through equation (\ref{matrix}). Let the principal factors  $W_t^1, \dots, W_t^m $ to be $m$ independent Brownian motions with different scale parameters $\sigma_i>0$ and  Laplace transform $\mathbb{E}(e^{-s W_t^i}) = e^{\frac12 \sigma_i^2 s^2 t}$.
We now only need to apply Theorem \ref{thm20}. By solving equation (\ref{relation3}) equal to zero we get, for $t\in [0,T]$,
\begin{equation}\label{Brownian3}
s^2t \sum_{j=1}^m d_j^2 \sigma_j^2 - \frac 12 s^2 t  \sum_{j=1}^m d_j^2 \sigma_j^2 + \ln\beta = 0,
\end{equation}
where $d_j$ is given in $(\ref{d_j})$.  Or equivalently,
\begin{equation}\label{Brownian4}
s^* = \left( \frac{-2 \ln\beta}{ t \sum_{j=1}^m d_j^2 \sigma_j^2}\right)^{\frac12}\;.
\end{equation}
Substituting  (\ref{Brownian4}) into equation (\ref{relation6}) at the point $u = (1,1,\dots,1)$ we can compute the value $K_t^i$ for $i=1,\dots,n$. That is,

\begin{equation}\label{Brownian5}
K_t^i = t^{\frac12} \left( \frac{-2 \ln\beta}{\sum_{j=1}^m\sigma_j^2  (\sum_{k=1}^n a_{kj})^2}\right)^{\frac12} \sum_{j=1}^m \sigma_j^2 a_{ij} \sum_{k=1}^n a_{kj}\;,
\end{equation}
for $t\in [0,T]$. Thus, the allocated capital to the $i^{th}$ department with respect to CEVa$R_{1-\beta}$ can be computed to be,
\begin{equation}\label{Brownian6}
L^i = \frac23 T^{\frac12} \left( \frac{-2 \ln\beta}{\sum_{j=1}^m\sigma_j^2  (\sum_{k=1}^n a_{kj})^2}\right)^{\frac12} \sum_{j=1}^m \sigma_j^2 a_{ij} \sum_{k=1}^n a_{kj}+ c^i\frac T2.
\end{equation}

Now as a special case, let the principal factors $W_t^1, \dots, W_t^m $ to be $m$ independent Brownian motions with common scale parameter $\sigma>0$ and common Laplace transform $\mathbb{E}(e^{-s W_t^i}) = e^{\frac12 \sigma^2 s^2 t}$. 
So, $(\ref{Brownian4})$ reduces to

\begin{equation}\label{Brownian1}
s^* = \left( \frac{-2 \ln\beta}{\sigma^2 t \sum_{j=1}^m d_j^2}\right)^{\frac12}\;,
\end{equation}
and the value $K_t^i$ is then, for $t\in [0,T]$, 

\begin{equation}\label{Brownian2}
K_t^i = \sigma t^{\frac12} \left( \frac{-2 \ln\beta}{\sum_{j=1}^m (\sum_{k=1}^n a_{kj})^2}\right)^{\frac12} \sum_{j=1}^m a_{ij} \sum_{k=1}^n a_{kj}\;.
\end{equation}
Thus, the allocated capital, $L^i$, to the $i^{th}$ department satisfying Definition \ref{capital for process} with respect to CEVa$R_{1-\beta}$ for this special case can be written as,
\begin{equation}\label{Brownian3}
L^i = \frac23 T^{\frac12} \sigma \left( \frac{-2 \ln\beta}{\sum_{j=1}^m (\sum_{k=1}^n a_{kj})^2}\right)^{\frac12} \sum_{j=1}^m a_{ij} \sum_{k=1}^n a_{kj} + c^i\frac T2.
\end{equation}

\subsection{Gamma Subordinator}
 Consider the general set-up defined through equation (\ref{matrix}). We let the principal factors $W_t^1, \dots, W_t^m $ to be $m$ independent gamma
processes with different parameters $\alpha_i, b_i > 0$ and Laplace transform 
\begin{equation}\label{19}
\mathbb{E}(e^{-s W_t^i}) = \left(1 + \frac sb_i\right)^{-\alpha_it} = \exp{\left[-t\alpha_i \ln \left( 1+ \frac sb_i\right) \right]}\;,\qquad s\geq 0\;.
\end{equation}



We now only need to apply Theorem \ref{thm20}. By solving equation (\ref{relation3}) equal to zero we get, for $t\in [0,T]$,
\begin{equation}\label{relation9}
t\sum_{j=1}^{m} \alpha_j \left(\ln(1+\frac{sd_j}{b_j}) -s \frac{d_j}{b_j+sd_j}\right) + \ln\beta = 0\;, \quad s>0\;,
\end{equation}
where $d_j$ is given in $(\ref{d_j})$. This is not as straight-forward as the equivalent equation for the previous example. Nonetheless, the value $s^*$ satisfying $(\ref{relation9})$ can be obtained numerically. Evaluating at the point $u = (1,1,\dots,1)$ and substituting into $(\ref{relation6})$ yields the capital allocation value $K_t^i$ for $i=1,\dots,n$. That is,

\begin{equation}\label{relation10}
K_t^i = -t\sum_{j=1}^m  a_{ij} \left(\frac{\alpha_j}{b_j+ s^* \sum_{k=1}^n a_{kj}}\right)\;,
\end{equation}
for $t\in [0,T]$ and where $s^*$ is the solution of equation (\ref{relation9}). Thus, the allocated capital to the $i^{th}$ department satisfying Definition \ref{capital for process} with respect to CEVa$R_{1-\beta}$ is given by,
\begin{equation}\label{relation10}
L^i = -a \frac T2\sum_{j=1}^m a_{ij} \left(\frac{\alpha_j}{b_j+ s^* \sum_{k=1}^n a_{kj}}\right) + c^i \frac T2 \;,
\end{equation}
for $1\leq i\leq n$.

\subsection{$\alpha$-stable Subordinator}
Consider the general set-up defined through equation (\ref{matrix}). We let the principal factors $W_t^1, \dots, W_t^m $ to be $m$ independent $\alpha$-stable processes with different parameter $\alpha_i\in (0,1)$ and Laplace transform 

\begin{equation}\label{alpha}
\mathbb{E}(e^{-s W_t^i}) = e^{-s^{\alpha_i}}\;,\qquad s\geq 0\;.
\end{equation}

We now only need to apply Theorem \ref{thm20}. By solving equation (\ref{relation3}) equal to zero we get the following equation, for $t\in [0,T]$ ,
\begin{equation}\label{alpha_2}
-t \sum_{j=1}^m\alpha_j (s d_j)^{\alpha_j} +t\sum_{j=1}^m  (s d_j)^{\alpha_j} + \ln\beta= 0\;,
\end{equation}

where $d_j$ is given in $(\ref{d_j})$.

Once again, this is not a straight-forward equation to solve. Nonetheless, the value $s^*$ satisfying $(\ref{alpha_2})$ can be obtained numerically. Evaluating at the point $u = (1,1,\dots,1)$ and substituting into $(\ref{relation6})$ yields the capital allocation value $K_t^i$ for $i=1,\dots,n$ and for all $t \in [0,T]$.

As a special case, consider the principal factors let $W_t^1, \dots, W_t^m $ to be $m$ independent $\alpha$-stable processes with common parameter $\alpha\in (0,1)$ and common Laplace transform 
\begin{equation}\label{17}
\mathbb{E}(e^{-s W_t^i}) = e^{-s^{\alpha}}\;,\qquad s\geq 0\;.
\end{equation}

We now only need to apply Theorem \ref{thm20}. By solving equation (\ref{relation3}) equal to zero we get, for $t\in [0,T]$ ,
$$-s^{\alpha}t\alpha \sum_{j=1}^m (d_j)^{\alpha} +s^{\alpha}t\sum_{j=1}^m  (d_j)^{\alpha} + \ln\beta= 0,$$
In this case, a solution can be readily obtained yielding,
\begin{equation}\label{relation7}
s^* =\left ( \frac{- \ln\beta}{(1-\alpha)t \sum_{j=1}^m d_j^{\alpha}}\right )^{\frac{1}{\alpha}}\;,
\end{equation}
for $t\in [0,T]$. Substituting (\ref{relation7}) into (\ref{relation6}) at the point $u = (1,1,\dots,1)$ yields the capital allocation value $K_t^i$ for $i=1,\dots,n$ at time $t\in [0,T]$. That is,
\begin{equation}\label{relation8}
K_t^i = -t^{\frac1{\alpha}}\alpha \left ( \frac{- \ln\beta}{(1-\alpha) \sum_{j=1}^m (\sum_{k=1}^n a_{kj})^{\alpha}}\right )^{\frac{\alpha -1}{\alpha}} \sum_{j=1}^m a_{ij} (\sum_{k=1}^n a_{kj})^{\alpha - 1}\;,
\end{equation}
for $1\leq i \leq n$. The capital allocation satisfying Definition \ref{capital for process} with respect to CEVa$R_{1-\beta}$ is given by, 
$$L^i = \frac 1T \int_0^T K_t^i dt + c^i\frac T2\;,$$
or more precisely,
\begin{equation}\label{relation8}
L^i = -\frac{\alpha^2}{\alpha +1} T^{\frac1{\alpha}} \left ( \frac{- \ln\beta}{(1-\alpha) \sum_{j=1}^m (\sum_{k=1}^n a_{kj})^{\alpha}}\right )^{\frac{\alpha -1}{\alpha}} \sum_{j=1}^m a_{ij} (\sum_{k=1}^n a_{kj})^{\alpha - 1}+c^i\frac{T}{2}\;.
\end{equation}



\begin{thebibliography}{99}
\bibitem{Ahmadi} Ahmadi-Javid A., {\it Entropic Value-at-Risk: A New Coherent Risk Measure}. Journal of Optimization Theory and Application, Springer. 2011.
\bibitem{Applebaum} Applebaum D. {\it L\' evy processes and Stochastic Calculus} (second edition), Cambridge University Press (2009)
\bibitem{Artzner} Artzner P.,  Delbaen F.,  Eber J.-M., and  Heath D. {\it Coherent measures of risk}. Math. Finance, 1999.

\bibitem{Asmussen}  Asmussen S. {\it Ruin probabilities}, volume 2 of Advanced Series on Statistical Science $\&$ Applied Probability. World Scientific Publishing Co. Inc., River Edge, NJ, 2000.
\bibitem{Assa1} Assa H. {\it Lebesgue property of risk measures for bounded c\' adl\' ag processes and applications.}, 2009.
\bibitem{Assa2} Assa H. {\it On Some Aspects of Coherent Risk Measures and their Applications} Ph.D.Thesis, 2011.

\bibitem{Billera} Billera, L. J.,  Heath, D. C., {\it Allocation of shared costs: a set of axioms
yielding a unique procedure}. Math. Oper. Res., 7(1):32--39, 1982.
\bibitem{Biffis_Morales} Biffis, E, Morales, M. {\it On a generalization of the Gerber-Shiu function to path-dependent penalties} Insurance: Mathematics and Economics.
Volume 46, Issue 1, Pages 92--97, 2009.
\bibitem{Biffis_Kyprianou} Biffis, E, Kyprianou, A. {\it A Note on Scale Functions and the Time Value of Ruin for L\' evy Insurance Risk Processes} Insurance: Mathematics and Economics.  46, no. 1, 85--91, 2010. 
\bibitem{Bowers} Bowers, N, Gerber, H, Hickman, J, Jones, D, Nesbitt, C.  {\it Actuarial Mathematics}. Society of Actuaries; 2nd edition, 1997.


\bibitem{Cheridito1} Cheridito P., Delbaen F., and Kupper M. {\it Coherent and convex monetary risk measures for bounded c\` adl\` ag processes.} Stochastic Process. Appl., 112(1):1--22, 2004.
\bibitem{Cheridito2} Cheridito P., Delbaen F., and Kupper M. {\it Coherent and convex monetary risk measures for unbounded c\` adl\` ag processes.} Finance Stoch., 9(3):369--387, 2005.
\bibitem {Cheridito3}  Cheridito P., Delbaen F., and Kupper M. {\it Dynamic monetary risk measures for bounded discrete-time processes}. Electron. J. Probab., 11:no. 3, 57--106 (electronic), 2006.

\bibitem{Cherny}   Cherny A. {\it Pricing with coherent risk}.   translation in Theory Probab. Appl. 52, no. 3, 2008.
\bibitem{Cramer} Cram\' er, H. {\it Historical review of Filip Lundberg's works on risk theory}. Scandinavian Actuarial Journal, 6--9,  1969.
\bibitem{Delbaen1} Delbaen F. {\it Coherent risk measures.} Cattedra Galileiana. [Galileo Chair]. Scuola Normale Superiore, Classe di Scienze, Pisa, 2000.
\bibitem{Delbaen2} Delbaen F. {\it Coherent risk measures on general probability spaces.} In Advances in finance and stochastics, pages 1--37. Springer, Berlin, 2002.
\bibitem{Delbaen3} Delbaen F. {\it Coherent Monetary Utility Functions}, Preprint, Available at  http://www.math.ethz.ch/$\sim$ delbaen under the name "Pisa Lecture Notes", 2005.


\bibitem{Fischer} Fischer T. {\it Risk Capital Allocation by Coherent Risk Measures Based on One-Sided Moments}, Insurance.: Math. Econ. 32(1), 135--146, 2003.
\bibitem{Follmer2} Follmer H. and Schied A. {\it Stochastic finance,  An introduction in discrete time.}, volume 27 of de Gruyter Studies in Mathematics. Walter de Gruyter $\&$ Co., Berlin, extended edition, 2004.

%
%

\bibitem{Furrer}
Furrer, H.J. Risk Processes Perturbed by a $\alpha$-stable L\'evy Motion.
{\it Scandinavian Actuarial Journal.}{\bf 1.} pp.~59--74, 1998.

\bibitem{Grandell}
Grandell, J. A Class of Approximations of Ruin Probabilities. {\it Scandinavian Actuarial Journal.} {\bf Supp}. pp. 37--52, 1977.

\bibitem{Kyprianou} Kyprianou, A.E. {\it Introductory Lecture Notes on Fluctuations of L\' evy processes with Applications}. Springer-Verlag, 2006.


\bibitem{Kalkbrener} Kalkbrener, M. {\it An axiomatic characterization of capital allocations of coherent risk measures}. Quantitative Finance, 9:8, 961--965, 2009.
%

%




\bibitem{Sato} Sato, K.-I. {\it L\' evy processes and Infinitely Divisible Distributions}. Cambridge University Press (1999).
%

\bibitem {Tasche3}  Tasche, D., {\it Euler allocation: theory and practice}. Technical document. Fitch Ratings, London, 2007.
\bibitem{Trufin}  Trufin, J.,  Albrecher, H., Denuit, M., {\it Properties of a Risk Measure Derived from Ruin Theory}. The geneva risk and insurance review palgrave, 36, 174--188, 2011.
\bibitem{Tsai} Tsai, C.C.L. and Willmot, G.E. {\it A generalized defective renewal equation for the surplus process perturbed by diffusion}. Insurance: Mathematics and Economics 30, 51--66, 2002.
\bibitem {Urban} Urban, M., Dittrich, J., Kluppelberg, C. and Stolting, R. {\it Allocation of risk capital to insurance portfolios}. Blatter der DGVFM,     Springer Berlin / Heidelberg. Pages 389--406, 2004.





\end{thebibliography}
\end{document}